\documentclass{llncs}

\usepackage{fancyhdr}
\usepackage{amssymb,amsmath}
\usepackage[usenames,dvipsnames]{color}
\usepackage[colorlinks=true,linkcolor=blue,citecolor=Mahogany]{hyperref}
\usepackage[poorman]{cleveref}
\usepackage[comma,authoryear,square,sort]{natbib} 
\usepackage{multirow}
\usepackage{mathtools}

\DeclarePairedDelimiter{\ceil}{\lceil}{\rceil}

\newcommand{\eps}{\varepsilon}
\renewcommand{\epsilon}{\eps}

\newcommand{\ignore}[1]{}
\newcommand{\E}[1]{\mathbb{E}\left[#1\right]}

\crefname{observation}{observation}{\bf Observation}
\begin{document}

\title{Sample Complexity for Winner Prediction in Elections}

\author{Arnab Bhattacharyya and Palash Dey\\ \texttt{\{arnabb,palash\}@csa.iisc.ernet.in}}

\institute{Department of Computer Science and Automation \\Indian Institute of Science - Bangalore, India.\\[10pt]Date: \today}

\maketitle
\pagestyle{fancy}


\begin{abstract}

Predicting the winner of an election is a favorite problem both for news media pundits and computational social choice theorists. Since it is often infeasible to elicit the preferences of all the voters in a typical prediction scenario, a common algorithm used for winner prediction is to run the election on a small sample of randomly chosen votes and output the winner as the prediction. We analyze the performance of this algorithm for many common voting rules.

More formally, we introduce the {\em $(\epsilon, \delta)$-winner determination problem}, where given an election on $n$ voters and $m$ candidates in which the margin of victory is at least $\epsilon n$ votes, the goal is to determine the winner with probability at least $1-\delta$. The margin of victory of an election is the smallest number of votes that need to be modified in order to change the election winner.  We show interesting lower and upper bounds on the number of samples needed to solve the $(\epsilon, \delta)$-winner determination problem for many common voting rules, including scoring rules, approval, maximin, Copeland, Bucklin, plurality with runoff, and single transferable vote. Moreover, the lower and upper bounds match for many common voting rules in a wide range of practically appealing scenarios. 
\end{abstract}


\newpage

\section{Introduction}

A common and natural way to aggregate preferences of agents is through
an {\em election}.  In a typical election, we have a set of candidates and a set of
voters, and each voter reports his preference about the candidates in the
form of a {\em vote}. We will assume that each vote is a
ranking of all the candidates. A {\em voting rule} selects one candidate as the winner
once all voters provide their votes. Determining the winner of an election is one of the most
fundamental  problems in social choice theory. 

In many situations, however, one wants to predict the winner without
holding the election for the entire population of voters. The
most immediate such example is an {\em election poll}. Here, the
pollster wants to quickly gauge public opinion in order to predict the
outcome of a full-scale election. For political elections,
exit polls (polls conducted on voters after they have
voted) are widely used by news media to predict the winner before 
official results are announced. In {\em surveys}, a full-scale
election is never conducted, and the goal is to determine the winner,
based on only a few sampled votes, for a hypothetical election on all the voters. For instance, it is
not possible to force all the residents of a city to fill out an
online survey to rank the local Chinese restaurants, and so only
those voters who do participate have their preferences aggregated. 

If the result of the poll or the survey has to reflect the true
election outcome, it is obviously necessary that the number of sampled
votes not be too small. Here, we investigate this fundamental
question: 

\begin{quote}
What is the minimum number of votes that need to be sampled
so that the winner of the election on the sampled votes is the
same as the winner of the election on all the votes?
\end{quote}

This question can be posed for any voting rule. The most immediate rule
to study is  the {\em plurality} voting rule, where each voter votes
for a single candidate and the candidate with most votes
wins. Although the plurality rule is the most common voting rule used
in political elections, it is important to extend the 
analysis to other popular voting rules. For example, the {\em single
  transferable vote} rule is used in political elections in Australia,
India and Ireland, and it was the subject of a nationwide referendum
in the UK in 2011. The {\em Borda} voting rule is used in the
Icelandic parliamentary elections. Outside politics, in private
companies and competitions, a wide variety of voting rules are
used. For example, the {\em approval} voting rule has been used by the
Mathematical Association of America, the American Statistical
Institute, and the Institute of Electrical and Electronics
Engineers, and {\em Condorcet consistent} voting rules are used
by many free software organizations. Section \ref{sec:prelim}
discusses the most common voting rules in use.

Regardless of the voting rule, though, the question of finding the
minimum number of vote samples required becomes trivial if a single
voter in the election can change the winning candidate. In this
case, all the votes need to be counted, because otherwise that single
crucial vote may not be sampled. We get around this problem by
assuming that in the elections we consider, the winning candidate wins
by a considerable {\em margin of victory}. Formally, the margin of
victory for an election is defined as the minimum number of votes that
must be changed in order to change the election winner. Note that the
margin of victory depends not only on the votes cast but also on the
voting rule used in the election.

\subsection{Our Contributions}

Let the number of voters be $n$ and the number of candidates $m$.
We introduce and study the following problem\footnote{Throughout this
  section, we use standard terminlogy from voting theory. For formal
  definitions, refer to Section \ref{sec:prelim}.}:
  
\begin{definition}($(\epsilon, \delta)$-winner determination)\\
 Given a voting rule and a set of $n$ votes over a set of $m$
 candidates such that the margin of victory is at least $\epsilon n$,  
 determine the winner of the election with probability at least
 $1-\delta$. (The probability is taken over the internal coin tosses of
 the algorithm.) 
\end{definition}

We remind the reader that there is no assumption about the
distribution of votes in this problem. Our goal is to solve the
$(\eps, \delta)$-winner determination problem by a randomized
algorithm that is allowed to query the votes of arbitrary voters. Each
query reveals the full vote of the voter. The minimum number of votes
queried by any algorithm that solves the $(\eps, \delta)$-winner
determination problem is termed the {\em   sample complexity}. The
sample complexity  can of course depend on $\eps$, $\delta$, $n$, $m$,
and the voting rule in use. 

A standard result \citep{canetti1995lower} shows that solving the above
problem for the majority rule on $2$ candidates requires at least
$\Omega(1/\eps^2  \log 1/\delta)$ samples (\Cref{thm:lb}). Also, a
straightforward argument (Theorem \ref{thm:gen}) using Chernoff bounds shows that for any homogeneous 
voting rule, the sample complexity is at most $O(m!^2/\eps^2 \cdot
\log(m!/\delta))$. So, when $m$ is a constant,  the sample complexity is of the order
$\Theta(1/\eps^2 \log 1/\delta)$ for any homogeneous voting rule that reduces to majority on
$2$ candidates (as is the case for all rules commonly used). Note that this bound is independent
of $n$ if $\eps$ and $\delta$ are constants, for any reasonable voting
rule!   

Our main technical contribution is in understanding the dependence of
the sample complexity on $m$, the number of candidates. Note
that the upper bound cited above has very bad dependence on $m$ and is
clearly unsatisfactory in situations when $m$ is large (such as in
online surveys about restaurants). 
\begin{itemize}
 \item We show that the sample complexity of the   $(\epsilon,
   \delta)$-winner determination problem is
   $\Theta(\frac{1}{\epsilon^2}\log \frac{1}{\delta})$  for the
   $k$-approval voting rule when $k=o(m)$ (\Cref{thm:kapp})  and the
   plurality with runoff voting rule (\Cref{thm:runoff}). In
   particular, for the plurality rule, the
   sample complexity is independent of $m$ as well as $n$!
\ignore{We also prove that, any $(\epsilon, \delta)$-winner
  determination algorithm for any voting that is same as the plurality
  voting rule for elections with two candidates has a sample
  complexity of $\Omega(\frac{1}{\epsilon^2}\log \frac{1}{\delta})$
  [\Cref{cor:lb}]. We remark that, all the commonly used voting rules
  that we know of including the ones that have been studied in this
  paper coincides with the plurality voting rule for elections with
  two candidates.}

 \item We show that the sample complexity of the $(\epsilon,
   \delta)$-winner determination problem is $O(\frac{\log(m/\delta)}{\epsilon^2})$ and $\Omega(\frac{\log
     m}{\epsilon^2} (1-\delta))$for the $k$-approval voting
   rule when $k=cm$ with $0<c<1$ (\Cref{thm:scr}),  Borda (\Cref{thm:strlwb}),
   approval (\Cref{thm:app}), maximin (\Cref{thm:maximin}), and
   Bucklin (\Cref{thm:bucklin}) voting rules. Note that when $\delta$
   is a constant, the upper and lower bounds match up to
   constants. We observe a surprising jump in the sample complexity of the $(\epsilon,
   \delta)$-winner determination problem by a factor of $\log m$ for the 
   $k$-approval voting rule as $k$ varies from $o(m)$ to $cm$ with $c\in(0,1)$.
 \item We show a sample complexity upper bound of $O(\frac{\log^3 \frac{m}{\delta}}{\epsilon^2})$ for the $(\epsilon,
   \delta)$-winner determination problem for the Copeland$^\alpha$ voting rule
   (\Cref{thm:copeland}) and $O(\frac{m^2(m+\log
     \frac{1}{\delta})}{\epsilon^2})$ for the STV  voting rule
   (\Cref{thm:stv}).
\end{itemize}

We summarize the results in Table 1.

\begin{table}[htbp]
  \begin{center}
  {\renewcommand{\arraystretch}{1.7}
 \begin{tabular}{|c|c|c| }\hline
  \textbf{Voting Rule}	& \multicolumn{2}{c|}{\textbf{Sample complexity}} \\\hline\hline
  $k$-approval	& $O(\frac{1}{\epsilon ^2}\log \frac{k}{\delta})$\Cref{thm:kapp} & $\Omega(\frac{\log (k+1)}{\epsilon^2}.\left( 1 - \delta \right))$\Cref{thm:strlwb}  \\\hline
  Scoring Rules	& \multirow{2}{*}{$O(\frac{\log \frac{m}{\delta}}{\epsilon^2})$\Cref{thm:scr}} &  \\\cline{1-1}
  Borda	&  & \multirow{5}{*}{$\Omega (
    \frac{\log m}{\epsilon^2}.\left( 1 - \delta \right) ) ^\dagger$ \Cref{thm:strlwb}} \\\cline{1-2}
  Approval	& $O(\frac{\log \frac{m}{\delta}}{\epsilon^2})$\Cref{thm:app} &  \\\cline{1-2}
  Maximin	& $O(\frac{\log \frac{m}{\delta}}{\epsilon^2})$\Cref{thm:maximin}&  \\\cline{1-2}
  Copeland	& $O(\frac{\log^3 \frac{m}{\delta}}{\epsilon^2})$\Cref{thm:copeland}&		\\\cline{1-2}
  Bucklin	& $O(\frac{\log \frac{m}{\delta}}{\epsilon^2})$\Cref{thm:bucklin}	& 	\\ \hline
  Plurality with runoff & $O(\frac{\log \frac{1}{\delta}}{\epsilon^2})$\Cref{thm:runoff}& \multirow{3}{*}{$\Omega(\frac{1}{\epsilon ^2}\log \frac{1}{\delta})^\ast$ \Cref{cor:lb}}\\\cline{1-2}
  STV		& $O(\frac{m^2(m+\log \frac{1}{\delta})}{\epsilon^2})$\Cref{thm:stv}&		\\\cline{1-2}
  Any homogeneous voting rule		& $O(\frac{m!^2 \log \frac{m!}{\delta}}{\epsilon^2})$\Cref{thm:gen}&		\\\hline
 \end{tabular}
 }
  \caption{\normalfont Sample complexity of the $(\epsilon,
    \delta)$-winner determination problem for various voting
    rules. $\dagger$--The lower bound of $\Omega ( \frac{\log
      m}{\epsilon^2}.\left( 1 - \delta \right) )$ also applies to any voting rule that is 
    Condorcet consistent. ${\ast}$-- The lower bound of $\Omega(\frac{1}{\epsilon ^2}\log
  \frac{1}{\delta})$ holds for any voting rule
  that reduces to the plurality voting rule for elections with two candidates.}
  \end{center}
\end{table} 

The rest of the paper is organized as follow. We introduce the terminologies and define the problem formally in Section \ref{sec:prelim}; we present the results on lower bounds in Section \ref{sec:lwb}; Section \ref{sec:upbd} contains the results on the upper bounds for various voting rules; finally, we conclude in Section \ref{sec:con}.

This paper is a significant extension of the conference version of this work~\cite{deysampling}: this extended version includes all the proofs.

\subsection{Related Work}

The subject of voting is at the heart of (computational) social choice
theory, and there is a vast amount of literature in this area.  Elections take place not
only in human societies but also in manmade social networks
\citep{boldi2009voting,rodriguez2007smartocracy} and, generally, in
many multiagent systems \citep{ephrati1991clarke,pennock2000social}. The winner determination
problem is the task of finding the winner in an election, given the
voting rule in use and the set of all votes cast. It is known that
there are natural voting rules, e.g., Kemeny's rule and Dodgson's
method, for which the winner determination problem is
\textsf{NP}-hard \citep{bartholdi1989voting,hemaspaandra2005complexity,hemaspaandra1997exact}. 

The general question of whether the outcome of an election can be
determined by less than the full set of votes is the subject of {\em
  preference elicitation}, a central category of problems in AI. The
$(\eps, \delta)$-winner determination problem also falls in this area
when the elections are restricted to those having margin of victory
at least $\eps n$. For general elections, the preference elicitation
problem was studied by Conitzer and Sandholm \citep{conitzer2002vote},
who defined an elicitation policy as an adaptive sequence of questions
posed to voters. They proved that finding an efficient elicitation
policy is \textsf{NP}-hard for many common voting rules. Nevertheless,
several elicitation policies have been developed in later work
\citep{conitzer2009eliciting, lu2011robust, lu2011vote, ding2012voting,
oren2013efficient} that
work well in practice and have formal guarantees under various
assumptions on the vote distribution. Another related work is that of
Dhamal and Narahari \citep{dhamal2013scalable} who show that if the
voters are members of a social network where neighbors in the network
have similar candidate votes, then it is possible to elicit the
votes of only a few voters to determine the outcome of the full
election. 

In contrast, in our work, we posit no assumption on the vote
distribution other than that the votes create a substantial margin of
victory for the winner. Under this assumption, we show that even for
voting rules in which winner determination is \textsf{NP}-hard in the
worst case, it is possible to sample a small number of votes to
determine the winner. Our work falls inside the larger framework of
{\em property testing} \citep{ron2001property}, a class of problems studied in theoretical
computer science, where the inputs are promised to either satisfy some
property or have a ``gap'' from instances satisfying the property. In
our case, the instances are elections which either have some candidate $w$ as the
winner or are ``far'' from having $w$ being the winner (in the sense
that many votes need to be changed).

The basic model of elections has been generalized in several other ways to capture
real world situations.  One important consideration is that the votes may be incomplete
rankings of the candidates and not a complete ranking. There can also
be uncertainty over which voters and/or candidates will eventually
turn up. The uncertainty may additionally come up from the voting rule that will be used eventually to select the winner. In these incomplete information settings, several winner models have been proposed, for example, robust winner~\citep{boutilier2014robust,lu2011robust,shiryaev2013elections}, 
multi winner~\citep{lu2013multi}, stable winner~\citep{falik2012coalitions}, approximate winner~\citep{doucetteapproximate}, 
probabilistic winner~\citep{bachrach2010probabilistic}. 
Hazon et al.~\citep{hazon2008evaluation} proposed useful methods to
evaluate the outcome of an election under various uncertainties. We do
not study the role of uncertainty in our work.

\paragraph{ Organization}
We formally introduce the terminologies in Section \ref{sec:prelim};
we present the results on lower bounds in Section \ref{sec:lwb};
Section \ref{sec:upbd}  contains the results on the upper bounds for
various voting rules;  finally, we conclude in Section \ref{sec:con}.

\section{Preliminaries}\label{sec:prelim}
\ignore{We denote the set $\{0,1, 2, \cdots\}$ by 
$\mathbb{N}$ and $\mathbb{N}^+ = \mathbb{N}\setminus\{0\}$. We denote the set $\{1, \cdots, k\}$ by $[k]$. 
For any two real numbers $s,t\in \mathbb{R}$, we denote the closed interval $\{x\in \mathbb{R} : s\le x\le t\}$ by $[s,t]$ 
and the open interval $\{x\in \mathbb{R} : s < x < t \}$ by $(s,t)$.
Let $\uplus$ denotes the disjoint union of sets. 
}

\subsection{Voting and Voting Rules}
Let $\mathcal{V}=\{v_1, \dots, v_n\}$ be the set of all \emph{voters} and $\mathcal{C}=\{c_1, \dots, c_m\}$ 
the set of all \emph{candidates}. 
Each voter $v_i$'s \textit{vote} is a  complete order over $\succ_i$ over the candidates $\mathcal{C}$. 
For example, for two candidates $a$ and $b$, $a \succ_i b$ means that the voter $v_i$ prefers $a$ to $b$. 
We denote the set of all complete orders over $\mathcal{C}$ by $\mathcal{L(C)}$. 
Hence, $\mathcal{L(C)}^n$ denotes the set of all $n$-voters' preference profiles $(\succ_1, \dots, \succ_n)$. 
\ignore{We denote the $(n-1)$-voters' preference profile $(\succ_1, \dots, \succ_{i-1}, \succ_{i+1}, \dots, \succ_n)$ 
by $\succ_{-i}$. }

A map $r:\uplus_{n,|\mathcal{C}|\in\mathbb{N}^+}\mathcal{L(C)}^n\longrightarrow \mathcal{C}$
is called a \emph{voting rule}. Given a vote profile $\succ \in
\mathcal{L}(\mathcal{C})^n$, we call $r(\succ)$ the {\em winner}. Note
that in this paper, each election has  a unique winner, and we ignore the
possibility of ties. A voting rule is called {\em homogeneous} if it selects the winner solely based on the 
fraction of times each complete order from $\mathcal{L(C)}$ appears as a vote in the election. All the commonly used voting rules 
including the ones that are studied in this paper are homogeneous.

\ignore{A map 
$t:2^\mathcal{C}\setminus\{\emptyset\}\longrightarrow \mathcal{C}$ is called a \emph{tie breaking rule}.
\ignore{A commonly used tie breaking rule is \emph{lexicographic} tie breaking rule where ties are broken 
according to a predetermined preference $\succ_t \in \mathcal{L(C)}$. }
A \emph{voting rule} is $r=t\circ r_c$, where $\circ$ denotes composition of mappings. 
\ignore{We restrict ourselves to the lexicographic tie-breaking rule
  in this work. }
In this paper, we do not take tie-breaking rules into account and, thus, will use the terms 
voting correspondence and voting rule interchangeably. Hence there can be 
more than one winners in case of a tie. We say a candidate wins uniquely, if it does not 
tie with any other candidate. If not mentioned explicitly, by winning we would mean winning uniquely. }

Given an election $E$, we can construct a weighted graph $G_E$ called 
\textit{weighted majority graph} from $E$. The set of vertices in $G_E$ is the set of candidates in $E$. 
For any two candidates $x$ and $y$, the weight on the edge $(x,y)$ is $D_E(x,y) = N_E(x,y) - N_E(y,x)$, 
where $N_E(x,y)(\text{respectively }N_E(y,x))$ is the number of voters who prefer $x$ to $y$ (respectively $y$ to $x$).\ignore{ We will call $D(x,y)$ 
to be the margin of victory of $x$ against $y$ in their pairwise election. }
A candidate $x$ is called the {\em Condorcet winner} in an election $E$ if $D_E(x,y) > 0$ for every other 
candidate $y \ne x$. A voting rule is called {\em Condorcet consistent} if it selects the Condorcet winner 
as the winner of the election whenever it exists. 
\ignore{A voting rule is called monotonic if improving winner's position in any vote without changing 
the ordering among other candidates does not change the winner.}

Some examples of common voting rules\footnote{In all these rules, the
  possibilities of ties exist. If they do happen, we assume that some
  arbitrary but fixed tie breaking rule is applied.} are:

\begin{itemize}

\item \textbf{Positional scoring rules:} A collection of $m$-dimensional vectors $\vec{s}_m=\left(\alpha_1,\alpha_2,\dots,\alpha_m\right)\in\mathbb{R}^m$ 
 with $\alpha_1\ge\alpha_2\ge\dots\ge\alpha_m$ and $\alpha_1>\alpha_m$ for every $m\in \mathbb{N}$ naturally defines a
 voting rule -- a candidate gets score $\alpha_i$ from a vote if it is placed at the $i^{th}$ position, and the 
 score of a candidate is the sum of the scores it receives from all the votes. 
 The winner is the candidate with maximum score. 

Without loss of generality, we assume that for any score vector $\vec{\alpha}$,
there exists a $j$ such that $\alpha_j =1$ and $\alpha_k = 0$ for all $k>j$.
 The vector $\alpha$ that is $1$ in the first $k$ coordinates and $0$
 otherwise gives the {\em $k$-approval}  voting rule. $1$-approval is
 called the {\em plurality} voting rule, and $(m-1)$-approval is
 called the {\em veto} voting rule. The score vector $(m-1, m-2,
 \dots, 1, 0)$ gives the {\em Borda} voting rule.
 
\item \textbf{Approval:} In approval voting, each voter approves a subset
 of candidates. The winner is the candidate which is approved by the
 maximum number of voters. 
 
\item \textbf{Maximin:} The maximin score of a candidate $x$ is $\min_{y\ne
   x} D(x,y)$. The winner is the candidate with maximum maximin
 score. 
 
\item \textbf{Copeland$^\alpha$:} The Copeland$^\alpha$ score of a candidate $x$ is $|\{y\ne x:D_{\mathcal{E}}(x,y)>0\}|+\alpha|\{y\ne x:D_{\mathcal{E}}(x,y)=0\}|$, where $\alpha\in [0,1]$. The winner is  the candidate with the maximum Copeland score. 
 
\item \textbf{Bucklin:} A candidate $x$'s Bucklin score is the minimum number $l$ such that more than half 
 of the voters rank $x$ in their top $l$ positions. The winner is the candidate with lowest Bucklin score.
 
\item \textbf{Plurality with runoff:} The top two candidates according to
 plurality score are selected first. The pairwise winner of these two
 candidates is selected as the winner of the election. This rule is
 often called the {\em runoff} voting rule.
 
\item \textbf{Single Transferable Vote:} In Single Transferable Vote (STV), 
 a candidate with least plurality score is dropped out of the election and its votes 
 are transferred to the next preferred candidate. If two or more
 candidates receive least plurality score,  then tie breaking rule is
 used. The candidate that remains after $(m-1)$ rounds is the winner.
\end{itemize}
 
Among the above voting rules, only the maximin and the Copeland voting
rules are Condorcet consistent. 

Given an election, the margin of victory of this election is:

\begin{definition}
 Given a voting profile $\succ$, the {\em margin of victory (\textsf{MOV})} is the smallest number of votes $k$ such that the winner can be changed by changing $k$ many votes in $\succ$, while keeping other votes unchanged.
\end{definition}

Xia \citep{xia2012computing} showed that for most common voting rules
(including all those mentioned above), when each voter votes
i.i.d. according to a distribution on the candidates, the margin of
victory is with high probability, either $\Theta(\sqrt{n})$ or
$\Theta(n)$.  

\subsection{Statistical Distance Measures}

Given a finite set $X$, a distribution $\mu$ on $X$ is defined as a function $\mu : X \longrightarrow [0,1]$, such that $\sum_{x\in X} \mu(x) = 1$. 
The finite set $X$ is called the base set of the distribution
$\mu$. We use the following distance measures among distributions in
our work.

\begin{definition}
 The {\em KL divergence}~\citep{kullback1951information} and the {\em Jensen-Shannon divergence}~\citep{lin1991divergence} between two distributions $\mu_1$ and $\mu_2$ on $X$ are defined as follows.
 \[ D_{KL}(\mu_1 || \mu_2) = \sum_{x\in X} \mu_1(x) \log \frac{\mu_1(x)}{\mu_2(x)} \]
 \[ JS(\mu_1, \mu_2) = \frac{1}{2} \left( D_{KL}\left(\mu_1 || \frac{\mu_1 + \mu_2}{2} \right) + D_{KL}\left(\mu_2 || \frac{\mu_1 + \mu_2}{2} \right) \right) \]
\end{definition}

The Jensen-Shannon divergence has subsequently been generalized to measure the mutual distance among more than two distributions as follows.

\begin{definition}
 Given $n$ distributions $\mu_1, \ldots, \mu_n$ over the same base
 set, the {\em generalized Jensen-Shannon divergence}\footnote{The
   generalized Jensen-Shannon divergence is often formulated with
   weights on each of the $n$ distributions. The definition here puts equal
 weight on each distribution and is sufficient for our purposes.}
among them is:
 \[ JS (\mu_1, \ldots, \mu_n) = \frac{1}{n} \sum_{i=1}^n D_{KL}\left(\mu_i || \frac{1}{n}\sum_{j=1}^n \mu_j\right) \]
\end{definition}

\subsection{Chernoff Bound}

We repeatedly use the following concentration inequality:

\begin{theorem}\label{thm:chernoff}
Let $X_1, \dots, X_\ell$ be a sequence of $\ell$ independent
random variables in $[0,1]$ (not necessarily identical). Let $S = \sum_i X_i$ and
let $\mu = \E{S}$. Then, for any $0 \leq \delta \leq 1$: 
$$\Pr[|S-\mu| \geq \delta \ell] < 2 \exp(-2\ell \delta^2)$$
and 
$$\Pr[|S - \mu| \geq \delta \mu] < 2\exp(-\delta^2\mu/3)$$
The first inequality is called an additive bound and the second
multiplicative. 
\end{theorem}

\section{Results on Lower Bounds}\label{sec:lwb}

Our lower bounds for the sample complexity of $(\eps, \delta)$-winner
determination are derived from information-theoretic lower bounds
for distinguishing distributions.

We start from the following basic observation.
Let $X$ be a random variable taking value $1$ with probability
$\frac{1}{2}-\epsilon$ and $0$ with probability
$\frac{1}{2}+\epsilon$; $Y$ be a random variable taking value $1$ with
probability $\frac{1}{2}+\epsilon$ and $0$ with probability
$\frac{1}{2}-\epsilon$.  Then, it is well-known that every algorithm needs
$\Omega(\frac{1}{\epsilon^2}\log \frac{1}{\delta})$ many samples to
distinguish between $X$ and $Y$ with probability of making an error
being at most $\delta$~\citep{canetti1995lower,bar2001sampling}. Immediately, we have:

\begin{theorem}\label{thm:lb}
 The sample complexity of the $(\epsilon, \delta)$-winner determination problem for the plurality voting rule is $\Omega(\frac{1}{\epsilon^2}\log \frac{1}{\delta})$.
\end{theorem}

\begin{proof}
 Consider an election with two candidates $a$ and $b$. Consider two
 vote distributions $X$ and $Y$. In $X$, exactly $\frac{1}{2} +
 \epsilon$ fraction of voters prefer $a$ to $b$ and thus $a$ is the
 plurality winner of the election. In $Y$, exactly $\frac{1}{2} +
 \epsilon$ fraction of voters prefer $b$ to $a$ and thus $b$ is the
 plurality winner of the election. Also, the margin of victory of both
 the elections corresponding to the vote distributions $X$ and $Y$ is
 $\epsilon n$, since each vote change can change the plurality score
 of any candidate by at most one. Any $(\epsilon, \delta)$-winner
 determination algorithm for plurality will give us a distinguisher between the
 distributions $X$ and $Y$ with probability of error at most $\delta$
 and hence will need $\Omega(\frac{1}{\epsilon^2}\log
 \frac{1}{\delta})$ samples. \qed
\end{proof}

\Cref{thm:lb} immediately gives us the following corollary.

\begin{corollary}\label{cor:lb}
 Every $(\epsilon, \delta)$-winner determination algorithm needs
 $\Omega(\frac{1}{\epsilon^2}\log \frac{1}{\delta})$ many samples for
 any voting rule which reduces to the plurality rule for two
 candidates. In particular, the lower bound holds for approval,
 scoring rules, maximin, Copeland,  Bucklin, plurality with runoff,
 and STV voting rules.  
\end{corollary}

\begin{proof}
All the voting rules mentioned in the statement except the approval voting rule is same as the plurality voting rule for elections with two candidates. 
Hence, the result follows immediately from \Cref{thm:lb} for the above voting rules except the approval voting rule. 
The result for the approval voting rule follows from the fact that any arbitrary plurality election is also a valid approval election where 
every voter approves exactly one candidate.\qed
\end{proof}

 We derive stronger lower bounds in terms of $m$ by explicitly viewing
 the $(\eps,\delta)$-winner determination problem as a {\em statistical
   classification} problem. In this problem, we are given a black box
 that contains a distribution $\mu$ which is guaranteed to be one of
 $\ell$ known distributions $\mu_1, \dots, \mu_\ell$. A {\em
   classifier} is a randomized oracle which has to determine the
 identity of $\mu$, where each oracle call produces a sample from
 $\mu$. At the end of its execution, the classifier announces a guess
 for the identity of $\mu$, which has to be correct with probability
 at least $1-\delta$. Using information-theoretic methods, Bar-Yossef
 \citep{bar2003sampling} showed the following:
 
\begin{lemma}\label{lem:kldiv}
 The worst case sample complexity $q$ of a classifier $C$ for $\mu_1,
 \ldots, \mu_\ell$ which does not make error with probability more
 than $\delta$ satisfies  following.
 \[  q \ge \Omega\left( \frac{\log \ell}{JS\left( \mu_1, \ldots, \mu_\ell \right)} . \left( 1 - \delta \right) \right) \]
\end{lemma}

The connection with our problem is the following.   A
set $V$ of $n$ votes on a candidate set $\mathcal{C}$ generates a
probability distribution $\mu_V$ on $\mathcal{L}(\mathcal{C})$, where 
$\mu_V(\succ)$ is proportional to the number of voters who voted
$\succ$. Querying a random vote from $V$ is then equivalent to
sampling from the distribution $\mu_V$. The margin of victory is proportional to the minimum statistical distance between $\mu_V$ and $\mu_W$,
over all the voting profiles $W$ having a different winner than the
winner of $V$. 

Now, suppose we have $m$ voting profiles $V_1, \dots, V_m$ having
different winners such that each $V_i$ has margin of 
victory at least $\eps n$. Any $(\eps,\delta)$-winner determination
algorithm must also be a statistical classifier for $\mu_{V_1}, \dots,
\mu_{V_m}$ in the above sense. It then remains to construct such
voting profiles for various voting rules which we do in the proof of
the following theorem:

\begin{theorem}\label{thm:strlwb}
 Every $(\epsilon, \delta)$-winner determination algorithm needs $\Omega \left( \frac{\log m}{\epsilon^2}.\left( 1 - \delta \right) \right)$ many samples for approval, Borda, Bucklin, and any Condorcet consistent voting rules, and $\Omega \left( \frac{\log k}{\epsilon^2}.\left( 1 - \delta \right) \right)$ many samples for the $k$-approval voting rule.
\end{theorem}

\begin{proof}
 For each voting rules mentioned in the theorem, we will show $d$ ($d=k+1$ for the $k$-approval voting rule and $d=m$ for the rest of the voting rules) many distributions $\mu_1, \ldots, \mu_d$ on the votes with the following properties the result follows from \Cref{lem:kldiv}. Let $V_i$ be an election where each vote $v\in \mathcal{L(C)}$ occurs exactly $\mu_i(v)\cdot n$ many times. Let $\mu = \frac{1}{d}\sum_{i=1}^d \mu_i$.
 \begin{enumerate}
  \item For every $i \ne j$, the winner in $V_i$ is different from the winner in $V_j$.
  \item For every $i$, the margin of victory of $V_i$ is $\Omega(\epsilon n)$.
  \item $D_{KL}(\mu_i || \mu) = O(\epsilon^2)$
 \end{enumerate}
 The distributions for different voting rules are as follows. Let the candidate set be $\mathcal{C} = \{ c_1, \ldots, c_m \}$. 
 
 \textbf{$k$-approval voting rule.} Fix any arbitrary $M:= k+1$ many candidates $c_1, \ldots, c_M$. For $i \in [M]$, we define a distribution $\mu_i$ on all $k$ sized subsets of $\mathcal{C}$ (for the $k$-approval voting rule, each vote is a $k$-sized subset of $\mathcal{C}$) as follows. Each $k$ sized subset corresponds to top $k$ candidates in a vote.
 $$ \mu_i(x) = \begin{cases}
                \frac{\epsilon}{{M-1 \choose k-1}} + \frac{1-\epsilon}{{M \choose k}}& \text{if } c_i\in x \text{ and } x\subseteq \{c_1, \ldots, c_{M}\}\\
                \frac{1-\epsilon}{{M \choose k}}& c_i\notin x \text{ and } x\subseteq \{c_1, \ldots, c_{M}\}\\
                0& \text{else}
               \end{cases}
 $$
 The score of $c_i$ in $V_i$ is $n\left( \eps + \left( 1 - \eps \right)\frac{{M-1 \choose k-1}}{{M \choose k}} \right)$, the
 score of any other candidate $c_j \in \{c_1, \ldots, c_M\}\setminus\{c_i\}$ is $n\left( 1 - \eps \right)\frac{{M-1 \choose k-1}}{{M \choose k}}$, and the score of the rest of the candidates is zero. Hence,
 the margin of victory is $\Omega(\epsilon n)$, since each
 vote change can reduce the score of $c_i$ by at most one and increase
 the score of any other candidate by at most one. This proves the
 result for the $k$-approval voting rule. Now, we show that $D_{KL}(\mu_i || \mu)$ to be $O(\epsilon^2)$. 
 \begin{eqnarray*}
  D_{KL}(\mu_i || \mu) 
  &=& \left( \eps+\left(1-\eps\right)\frac{k}{M} \right)\log\left( 1-\eps+\eps\frac{M}{k} \right) + \left( 1-\eps \right)\left( 1-\frac{k}{M} \right)\log\left(1-\eps\right)\\
  &\le& \left( \eps+\left(1-\eps\right)\frac{k}{M} \right)\left( \eps\frac{M}{k} - \eps \right) - \left( 1-\eps \right)\left( 1-\frac{k}{M} \right)\eps\\
  &=& \eps^2 \left( \frac{M}{k}-1 \right)\\
  &\le& 2\eps^2
 \end{eqnarray*}
 
 \textbf{Approval voting rule.} The result follows from the fact that every $\frac{m}{2}$-approval election is also a valid approval election and \Cref{lem:approval}.
 
 \textbf{Borda, any Condorcet consistent voting rule.} The score vector for the Borda voting rule which we use in this proof is $(m, m-1, \ldots, 1)$. For $i \in [m]$, we define a distribution $\mu_i$ on all possible linear orders over $\mathcal{C}$ as follows.
 $$ \mu_i(x) = \begin{cases}
                \frac{2\epsilon}{m!} + \frac{1-\epsilon}{m!}& \text{if } c_i \text{ is within top } \frac{m}{2} \text{ positions in }x.\\
                \frac{1-\epsilon}{m!}& \text{else}
               \end{cases}
 $$
 The score of $c_i$ in $V_i$ is $\frac{mn}{2}(1+\frac{\epsilon}{2})$ whereas the score of any other candidate $c_j \ne c_i$ is $\frac{mn}{2}$. Hence, the margin of victory is at least $\frac{\epsilon n}{8}$, since each vote change can reduce the score of $c_i$ by at most $m$ and increase the score of any other candidate by at most $m$. Also, in the weighted majority graph for the election $V_i$, $w(c_i, c_j) = \frac{\epsilon n}{2}$. Hence, the margin of victory is at least $\frac{\epsilon n}{4}$, since each vote change can change the weight of any edge in the weighted majority graph by at most two. Now, we show that $D_{KL}(\mu_i || \mu)$ to be $O(\epsilon^2)$.
 \begin{eqnarray*}
  D_{KL}(\mu_i || \mu) 
  &=& \frac{1+\eps}{2}\log\left( 1+\eps \right) + \frac{1-\eps}{2}\log\left( 1-\eps \right)\\
  &\le& \frac{1+\eps}{2}\eps - \frac{1-\eps}{2}\eps\\
  &=& \eps^2
 \end{eqnarray*}
 
 \textbf{Bucklin.} For $i \in [m]$, we define a distribution $\mu_i$ on all $\frac{m}{4}$ sized subsets of $\mathcal{C}$ as follows. Each $\frac{m}{4}$ sized subset corresponds to the top $\frac{m}{4}$ candidates in a vote.
 $$ \mu_i(x) = \begin{cases}
                \frac{1-\epsilon}{{m-1 \choose \frac{m}{4}-1}} + \frac{\epsilon}{{m \choose \frac{m}{4}}}& \text{if } c_i\in x\\
                \frac{\epsilon}{{m \choose \frac{m}{4}}}& \text{else}
               \end{cases}
 $$
 The candidate $c_i$ occurs within the top $\frac{m}{4}$ positions at least $n(1-\frac{3\epsilon}{4})$ many times, and any candidate $c_j \ne c_i$ occurs within the top $\frac{m}{4}$ positions at most $\frac{n}{3} - \frac{\epsilon n}{12}$ many times. Hence, the margin of victory is at least $\frac{\epsilon n}{6}$, since each vote change can change the number of time any particular candidate occurs within top $\frac{m}{4}$ positions by at most one. Now, we show that $D_{KL}(\mu_i || \mu)$ to be $O(\epsilon^2)$. 
 \begin{eqnarray*}
  D_{KL}(\mu_i || \mu) 
  &=& \left( 1-\frac{3\eps}{4} \right)\log\left(4-3\eps\right) + \frac{3\eps}{4}\log \eps\\
  &\le& \left( 1-\frac{3\eps}{4} \right)\log\left(4-3\eps\right)\\
  &=& 2\eps^2
 \end{eqnarray*}\qed
\end{proof}

\section{Results on Upper Bounds}\label{sec:upbd}

In this section, we present the upper bounds on the sample complexity
of the $(\epsilon, \delta)$-winner determination problem for various
voting rules. The general framework for proving the upper bounds is as
follows. For each voting rule, we first prove a useful structural
property about the election when the margin of victory is known to be
at least $\epsilon n$. Then, we sample a few votes uniformly
at random to estimate either the  score of the candidates for score
based voting rules or weights of the edges in the weighted majority
graph for other voting rules. Finally,  appealing to the structural
property that has been established, we argue that, the winner of the
election on the sampled votes will be the  same as the winner of the
election, if we are able to estimate either the scores of the candidates or
the weights of the edges in the weighted majority graph to a certain level of accuracy.  

Before getting into specific voting rules, we prove a straightforward
bound on the sample  complexity for the $(\epsilon, \delta)$-winner
determination problem for {\em any} voting rule.  

\begin{theorem}\label{thm:gen}
 There is a $(\epsilon, \delta)$-winner determination algorithm for every homogeneous voting rules with sample complexity $O(\frac{m!^2 \log \frac{m!}{\delta}}{\epsilon^2})$.
\end{theorem}

\begin{proof}
 We sample $\ell$ votes uniformly at random from the set of votes with
 replacement. Let $X_i$ be an indicator random variable that is $1$
 exactly when $x$ is the $i$'th sample, and let $g(x)$ be the total
 number of voters whose vote is $x$. Define $\hat{g}(x) =
 \frac{n}{l}\sum_{i=1}^{l}X_i$.  Using the Chernoff bound (\Cref{thm:chernoff}), we have the following:
 $$ \Pr\left[ |\hat{g}(x) - g(x)| > \frac{\epsilon n}{2m!} \right] \le
 2\cdot {\exp\left(- \frac{\eps^2 \ell}{2m!^2}\right)}$$
 By using the union bound, we have the following,
 \begin{eqnarray*}
  \Pr\left[\exists x\in \mathcal{L(C)}, |\hat{g}(x) - g(x)| >
    \frac{\epsilon n}{2m!} \right] &\le& 2m!\cdot \exp\left(-\frac{\eps^2 \ell}{2m!^2}\right)
 \end{eqnarray*}
 Since the margin of victory is $\epsilon n$ and the voting rule is anonymous, the winner of the $\ell$
 sample votes will be same as the winner of the election if
 $|\hat{g}(x) - g(x)| \le  \frac{\epsilon n}{2m!}$ for every linear
 order $x\in \mathcal{L(C)}$. Hence, it is enough to take $\ell =
 O(m!^2/\eps^2 \cdot \log(m!/\delta))$. \qed
\end{proof}

\subsection{Approval Voting Rule}

We derive the upper bound on the sample complexity for the $(\epsilon, \delta)$-winner determination problem for the approval voting rule.

\begin{lemma}\label{lem:approval}
 If $ \textsf{MOV} \ge \epsilon n $ and $w$ be the winner of a approval election, then, 
 $ s(w) - s(x) \ge \epsilon n, $ for every candidate $ x \ne w $, where $s(y)$ is the number of approvals that a candidate $y$ receives.
\end{lemma}

\begin{proof}
 Suppose there is a candidate $ x \ne w $ such that $ s(w) - s(x) < \epsilon n$. Then there must exist 
 $\epsilon n - 1$ votes which does not approve the candidate $x$. We modify these votes to make it 
 approve $x$. This makes $w$ not the unique winner in the modified election. This contradicts the fact that the \textsf{MOV} 
 is at least $\epsilon n$.\qed
\end{proof}

\begin{theorem}\label{thm:app}
 There is a $(\epsilon, \delta)$-winner determination algorithm for the approval voting rule with sample complexity $O(\frac{\log({m}/{\delta})}{\epsilon^2})$.
\end{theorem}

\begin{proof}
Suppose $w$ is the winner. We sample $\ell$ votes uniformly at random
from the set of votes with replacement. For a candidate $x$, let $X_i$
be a random variable  indicating whether the $i$'th vote sampled
approved $x$. Define $\hat{s}(x) =
\frac{n}{l}\sum_{i=1}^{l}X_i$. Then,  by an argument analogous to the
proof of \Cref{thm:gen}, $\Pr[\exists x \in \mathcal{C}, |\hat{s}(x)
-s(x)| >\eps n/2]\leq  2m\cdot \exp\left(-\eps^2\ell/2\right)$. Thus
since \textsf{MOV}$\geq \eps n$ and by \Cref{lem:approval}, if we  take $\ell = O(\frac{\log m/\delta}{\eps^2})$, $\hat{s}(w)$ is
greater than $\hat{s}(x)$ for  all $x \neq w$.\qed
\end{proof}

\subsection{Scoring Rules}

Now, we move on to the scoring rules. Again, we first establish a
structural consequence of having large \textsf{MOV}.

\begin{lemma}\label{lem:scr}
Suppose $\alpha = (\alpha_1, \dots, \alpha_m)$ be a normalized score
vector and $w$ is the winner of an election using scoring rule
$\alpha$ with \textsf{MOV} $ \ge \epsilon n$.
Then, $ s(w) - s(x) \ge \alpha_1 \epsilon n/2$ for every candidate $ x \ne w $, where $s(w)$ and $s(x)$ denote the score of the candidates $w$ and $x$ respectively.
\end{lemma}

\begin{proof}
 There must be at least $\epsilon n$ many votes
 where $w$ is preferred over $x$, since we can make $x$ win the
 election by exchanging the positions of $x$ and $w$ in all these
 votes and $\textsf{MOV} \geq \epsilon n$. Let $v$ be a vote where $w$ is
 preferred to $x$. Suppose we replace the vote $v$ by another vote $v^{\prime}
 = x\succ  \text{others}\succ w$. We claim that this replacement
 reduces the current value of $s(w)-s(x)$ by at least $\alpha_1$. If
 we change $\eps n/2$ such votes, then $s(w)-s(x)$ decreases by at
 least $\alpha_1\eps n/2$ but, at the same time, $w$ must still be
 the winner after the vote changes because of the \textsf{MOV} condition. So,
 $s(w)-s(x)\geq \alpha_1 \eps  n/2$. 

To prove the claim, suppose $w$ and $x$ were receiving a score of
$\alpha_i$ and $\alpha_j$ respectively from the vote $v$. By replacing
the vote $v$ by $v^{\prime}$, the current value of $s(w)-s(x)$ reduces
by $\alpha_1-\alpha_j+\alpha_i$, since $\alpha_m=0$. Now,
$\alpha_1-\alpha_j+\alpha_i \ge \alpha_1$ since in the vote $v$, the
candidate $w$ is preferred over $x$ and hence,
$\alpha_j<\alpha_i$. This proves the result. \qed
\end{proof}

\begin{theorem}\label{thm:scr}
 Suppose $\alpha = (\alpha_1, \dots, \alpha_m)$ be a normalized score
 vector. There is a $(\epsilon, \delta)$-winner determination
 algorithm for the  $\alpha$-scoring rule  with sample complexity
 $O(\frac{\log({m}/{\delta})}{\epsilon^2})$. 
\end{theorem}

\begin{proof}
 It is enough to show the result for the $(2\epsilon, \delta)$-winner determination problem. 
We sample $\ell$ votes uniformly
 at random from the set of votes with replacement. For a candidate
 $x$, define $X_i = \frac{\alpha_i}{\alpha_1}$ if
 $x$ gets a score of $\alpha_i$ from the  $i${th} sample vote, and let
 $\hat{s}(x) = \frac{n\alpha_1}{\ell}\sum_{i=1}^\ell X_i$.  Now, using
 Chernoff bound (\Cref{thm:chernoff}), we have:
 $$ \Pr\left[ \left|\hat{s}(x) - s(x)\right| \ge \alpha_1 \epsilon n/4\right] \le
 2\exp\left(-\frac{\epsilon^2  \ell}{2}\right)$$
The rest of the proof follows from  an argument analogous to the proof of
 \Cref{thm:app} using \Cref{lem:scr}. \qed
\end{proof}

From \Cref{thm:scr}, we have a $(\epsilon, \delta)$-winner
determination algorithm for the $k$-approval voting rule which needs
$O(\frac{\log({m}/{\delta})}{\epsilon^2})$ many samples for any $k$. This is
tight by \Cref{thm:strlwb} when $k = cm$ for some $c\in (0,1)$.

When $k=o(m)$, we have a lower bound of
$\Omega(\frac{1}{\epsilon^2}\log \frac{1}{\delta})$ for the
$k$-approval voting rule (see \Cref{cor:lb}). We show next that this
lower bound is also tight for the $k$-approval voting rule when
$k=o(m)$. Before embarking  on the proof of the above fact, we prove
the following lemma which will be crucially used.  

\begin{lemma}\label{lem:funmax}
 Let $f : \mathbb{R} \longrightarrow \mathbb{R}$ be a function defined by $f(x) = e^{-\frac{\lambda}{x}}$. Then, 
 \[  f(x) + f(y) \le f(x+y), \text{ for } x,y > 0, \frac{\lambda}{x+y} > 2, x < y \]
\end{lemma}

\begin{proof} For the function $f(x)$, we have following.
 \begin{eqnarray*}
  f(x) &=& e^{-\frac{\lambda}{x}} \\
  \Rightarrow f^{\prime}(x) &=& \frac{\lambda}{x^2} e^{-\frac{\lambda}{x}}\\
  \Rightarrow f^{\prime\prime}(x) &=& \frac{\lambda^2}{x^4} e^{-\frac{\lambda}{x}} - \frac{2\lambda}{x^3} e^{-\frac{\lambda}{x}}
 \end{eqnarray*}
 Hence, for $x,y > 0, \frac{\lambda}{x+y} > 2, x < y$ we have $f^{\prime\prime}(x), f^{\prime\prime}(y), f^{\prime\prime}(x+y) > 0$. This implies following for $x < y$ and an infinitesimal positive $\delta$.
 \begin{eqnarray*}
  f^{\prime}(x) &\le& f^{\prime}(y)\\
  \Rightarrow \frac{f(x-\delta) - f(x)}{\delta} &\ge& \frac{ f(y) - f(y-\delta)}{\delta} \\
  \Rightarrow f(x) + f(y) &\le& f(x-\delta) + f(y+\delta)\\
  \Rightarrow f(x) + f(y) &\le& f(x+y)
 \end{eqnarray*}\qed
\end{proof}

\begin{theorem}\label{thm:kapp}
 There is a $(\epsilon, \delta)$-winner determination algorithm for the $k$-approval voting rule  with sample complexity $O(\frac{\log(\frac{k}{\delta})}{\epsilon^2})$.
\end{theorem}

\begin{proof}
It is enough to show the result for the $(2\epsilon, \delta)$-winner determination problem. 
We sample $\ell$ votes uniformly at random from the set of votes with
replacement. For a candidate $x$, let $X_i$ be a random variable
indicating whether $x$ is among the top $k$ candidates for the $i^{th}$
vote sample. Define $\hat{s}(x) = \frac{n}{\ell}\sum_{i=1}^{l}X_i$,
and let $s(x)$ be the actual score of $x$. Then by the multiplicative Chernoff bound
(\Cref{thm:chernoff}), we have:
 $$ \Pr\left[ |\hat{s}(x) - s(x)| > \epsilon n \right] \le
 2\exp\left(-\frac{\epsilon^2 \ell n}{3 s(x)}\right)$$
 By union bound, we have the following,
 \begin{eqnarray*}
  && \Pr[ \exists x\in \mathcal{C}, |\hat{s}(x) - s(x)| > \epsilon n ]\\
  &\le& \sum_{x\in \mathcal{C}} 2\exp\left(-\frac{\epsilon^2 \ell n}{3
      s(x)}\right) \\
  &\le& 2k\exp\left(-\epsilon^2 \ell/3\right)
 \end{eqnarray*}
 Let the candidate $w$ be the winner of the election. The second
 inequality in the above derivation follows from the fact that, the
 function $\sum_{x\in \mathcal{C}} {\exp\left(-\frac{\epsilon^2 \ell n}{3 s(x)}\right)}$ is maximized in the domain, defined by the constraint: for every candidate $x\in \mathcal{C}$, $s(x) \in [0,n]$ and $\sum_{x\in\mathcal{C}} s(x) = kn$, by setting $s(x)=n$ for every $x \in \mathcal{C}^\prime$ and $s(y)=0$ for every $y \in \mathcal{C}\setminus\mathcal{C}^\prime$, for any arbitrary subset $\mathcal{C}^\prime \subset \mathcal{C}$ of cardinality $k$ (due to \Cref{lem:funmax}). The rest of the proof follows by an argument analogous to the proof of \Cref{thm:gen} using \Cref{lem:scr}.\qed
\end{proof}

Notice that, the sample complexity upper bound in \Cref{thm:kapp} is independent of $m$ for the plurality voting rule. \Cref{thm:kapp} in turn implies the following Corollary which we consider to be of independent interest.
\begin{corollary}\label{cor:linfty}
 There is an algorithm to estimate the $\ell_\infty$ norm $\ell_\infty(\mu)$ of a distribution $\mu$ within an additive factor of $\eps$ by querying only $O(\frac{1}{\eps^2} \log \frac{1}{\delta})$ many samples, if we are allowed to get i.i.d. samples from the distribution $\mu$.
\end{corollary}
Such a statement seems to be folklore in the statistics community \citep{dvoretzky1956asymptotic}. Recently in an independent and nearly simultaneous work, Waggoner \citep{Waggoner2015} obtained a sharp bound of $\frac{4}{\eps^2} \log(\frac{1}{\delta})$ for the sample complexity in \Cref{cor:linfty}.

\subsection{Maximin Voting Rule}

We now turn our attention to the maximin voting rule. The idea is to sample enough numer of votes such that 
we are able to estimate the weights of the edges in the weighted majority graph with certain level of accuracy which in turn 
leads us to predict winner.

\begin{lemma}\label{lem:maximin}
 Suppose $ \textsf{MOV} \ge \epsilon n $ and $w$ be the winner of a maximin election. Then, 
 $ s(w) - s(x) \ge \epsilon n, $ for every candidate $ x \ne w $, where $s(.)$ is the maximin score.
\end{lemma}

\begin{proof}
 Let $w$ be the winner and $x$ be any arbitrary candidate other than $w$. Suppose, for contradiction, $ s(w) - s(x) < \epsilon n $. Suppose $y$ be a candidate such that $N(w,y) = s(w)$. Now there exist at least $\epsilon n - 1$ votes as below.
 $$ c_1\succ \dots \succ w\succ \dots \succ y\succ \dots \succ c_{m-2} $$
 We replace $\epsilon n - 1$ of such votes by the votes as below.
 $$ c_1\succ \dots \succ y\succ \dots \succ c_{m-2}\succ w $$
 This makes the maximin score of $w$ less than the maximin score of $x$. This contradicts the assumption that $\textsf{MOV} \ge \epsilon n$.\qed
\end{proof}

\begin{theorem}\label{thm:maximin}
 There is a $(\epsilon, \delta)$-winner determination algorithm for the maximin voting rule  with sample complexity $O(\frac{\log \frac{m}{\delta}}{\epsilon^2})$.
\end{theorem}

\begin{proof}
 Let $x$ and $y$ be any two arbitrary candidates. We sample $\ell$
 votes uniformly at random from the set of votes with
 replacement. Let $X_i$ be a random variable  defined as follows.
 $$ X_i = \begin{cases}
           1,& \text{if } x\succ y \text{ in the } i^{th} \text{ sample}\\
           -1,& \text{else}
          \end{cases}
 $$
Define $\hat{D}(x,y) = \frac{n}{l}\sum_{i=1}^{l}X_i$. 
 We estimate $\hat{D}(x,y)$ within the closed ball of radius $\epsilon
 n/2$ around $D(x,y)$ for every candidates $x, y\in \mathcal{C}$ and
 the rest of the proof  follows from  by an argument analogous to the
 proof of \Cref{thm:app} using \Cref{lem:maximin}. \qed
\end{proof}

\subsection{Copeland Voting Rule}

Now, we move on to the Copeland$^\alpha$ voting rule. The approach is similar
to the maximin voting rule. However, it turns out that we need to
estimate the edge weights of the weighted majority graph more
accurately for the Copeland$^\alpha$ voting rule. Xia introduced the brilliant quantity called the {\em relative margin of victory} (see Section 5.1 in \citep{xia2012computing}) which will be used crucially for showing sample complexity upper bound for the Copeland$^\alpha$ voting rule. Given an election, a candidate $x\in C$, and an integer (may be negative also) $t$, $s^\prime_t(V, x)$ is defined as follows. 

$$s^\prime_t(V, x) = |\{ y\in C: y\ne x, D(y,x)<2t \}| + \alpha|\{ y\in C: y\ne x, D(y,x)=2t \}|$$

For every two distinct candidates $x$ and $y$, the relative margin of victory, denoted by $RM(x,y)$, between $x$ and $y$ is defined as the minimum integer $t$ such that, $s^\prime_{-t}(V, x) \le s^\prime_t(V, y)$. Let $w$ be the winner of the election $\mathcal{E}$. We define a quantity $\Gamma(\mathcal{E})$ to be $\min_{x\in C\setminus\{w\}} \{RM(w,x)\}$. Notice that, given an election $\mathcal{E}$, $\Gamma(\mathcal{E})$ can be computed in polynomial amount of time. Now we have the following lemma.

\begin{lemma}\label{lem:copeland}
Suppose $ \textsf{MOV} \ge \epsilon n $ and $w$ be the winner of a Copeland$^\alpha$ election. Then, 
 $ RM(w, x) \ge \frac{\epsilon n}{2(\ceil*{\log m} +1)}, $ for every candidate $ x \ne w $.
\end{lemma}

\begin{proof}
 Follows from Theorem 11 in \citep{xia2012computing}.
\end{proof}

\begin{theorem}\label{thm:copeland}
 There is a $(\epsilon, \delta)$-winner determination algorithm for Copeland$^\alpha$ voting rule with sample complexity $O(\frac{ \log^3 \frac{m}{\delta}}{\epsilon^2})$.
\end{theorem}

\begin{proof}
 Let $x$ and $y$ be any two arbitrary candidates and $w$ the Copeland$^\alpha$ winner of the election. We estimate $D(x,y)$ within the closed ball of radius $\frac{\epsilon n}{5(\ceil*{\log m} +1)}$ around $D(x,y)$ for every candidates $x, y\in \mathcal{C}$ in a way analogous to the proof of \Cref{thm:maximin}. This needs $O(\frac{ \log^3 \frac{m}{\delta}}{\epsilon^2})$ many samples. The rest of the proof follows from \Cref{lem:copeland} by an argument analogous to the proof of \Cref{thm:gen}. \qed
\end{proof}

\subsection{Bucklin Voting Rule}

For the Bucklin voting rule, we will estimate how many times each
candidate occurs within the first 
$k$ position for every $k\in [m]$. This eventually leads us to predict the winner of the election due to the following lemma.

\begin{lemma}\label{lem:bucklin}
 Suppose \textsf{MOV} of a Bucklin election be at least $\epsilon n$. Let $w$ be the winner of the election and $x$ be any 
 arbitrary candidate other than $w$. Suppose
 $$ b_w = \min_i \{ i : w \text{ is within top i places in at least } \frac{n}{2} + \frac{\epsilon n}{3} \text{ votes} \} $$
 $$ b_x = \min_i \{ i : x \text{ is within top i places in at least } \frac{n}{2} - \frac{\epsilon n}{3} \text{ votes} \} $$
 Then, $ b_w < b_x $.
\end{lemma}

\begin{proof}
 We prove it by contradiction. So, assume $ b_w \ge b_x $. Now by changing $\frac{\epsilon n}{3}$ votes, 
 we can make the Bucklin score of $w$ to be at least $b_w$. By changing another $\frac{\epsilon n}{3}$ votes, we can make the Bucklin score of $x$ to be at most $b_x$. Hence, by changing $\frac{2 \epsilon n}{3}$ votes, it is possible not to make $w$ the unique winner which contradicts the fact that the \textsf{MOV} is at least $\epsilon n$.\qed
\end{proof}

\begin{theorem}\label{thm:bucklin}
 There is a $(\epsilon, \delta)$-winner determination algorithm for Bucklin voting rule with sample complexity $O(\frac{\log \frac{m}{\delta}}{\epsilon^2})$.
\end{theorem}

\begin{proof}
 Let $x$ be any arbitrary candidate and $1\le k\le m$. We sample $l$ votes uniformly at random from the set of votes with replacement. Let $X_i$  be a random variable defined as follows.
 $$ X_i = \begin{cases}
           1,& \text{if } x \text{ is within top } k \text{ places in } i^{th} \text{ sample}\\
           0,& \text{else}
          \end{cases}
 $$
 Let $\hat{s}_k(x)$ be the estimate of the number of times the candidate $x$ has been placed within top $k$ positions. 
 That is, $\hat{s}_k(x) = \frac{n}{l} \sum_{i=1}^{l} X_i$. Let
 $s_k(x)$ be the number of times the candidate $x$ been placed in top
 $k$ positions. Clearly, $E[\hat{s}_k(x)] = \frac{n}{\ell} \sum_{i=1}^{\ell} E[X_i] = s_k(x) $.
 We estimate $\hat{s}_k(x)$ within the closed ball of radius $\epsilon
 n/2$ around $s_k(x)$ for every candidate $x \in \mathcal{C}$ and
 every integer $k\in [m]$,  and the rest of the proof follows from  by
 an argument analogous to the proof of \Cref{thm:app} using \Cref{lem:bucklin}.\qed
\end{proof}

\subsection{Plurality with Runoff Voting Rule}

Now, we move on to the plurality with runoff voting rule. In this case, we first estimate the plurality score of each of the candidates. In the next round, we estimate the pairwise margin of victory of the two candidates that qualifies to the second round.

\begin{lemma}\label{lem:runoff}
 Suppose $ \textsf{MOV} \ge \epsilon n $, and $w$ and $r$ be the winner and runner up of a plurality with runoff election respectively, and $x$ be any arbitrary candidate other than and $r$. Then, following holds. Let $s(.)$ denote plurality score of candidates. Then following holds.
 \begin{enumerate}
  \item $D(w,r) > 2 \epsilon n$.
  \item For every candidate $x \in \mathcal{C}\setminus\{w,r\}$, $ 2 s(w) > s(x) + s(r) +\epsilon n$.
  \item If $s(x) > s(r) - \frac{\epsilon n}{2}$, then $D(w,x) > \frac{\epsilon n}{2}$.
 \end{enumerate}
\end{lemma}

\begin{proof}
 If the first property does not hold, then by changing $\epsilon n$ votes, we can make $r$ winner.
 If the second property does not hold, then by changing $\epsilon n$ votes, we can make both $x$ and $r$ qualify to the second round.
 If the third property does not hold, then by changing $\frac{\epsilon n}{2}$ votes, the candidate $x$ can be sent to the second round of the runoff election. By changing another $\frac{\epsilon n}{2}$ votes, $x$ can be made to win the election. This contradicts the \textsf{MOV} assumption. \qed
\end{proof}

\begin{theorem}\label{thm:runoff}
 There is a $(\epsilon, \delta)$-winner determination algorithm for the plurality with runoff voting rule with sample complexity $O(\frac{\log \frac{1}{\delta}}{\epsilon^2})$.
\end{theorem}

\begin{proof}
 Let $x$ be any arbitrary candidate. We sample $l$ votes uniformly at random from the set of votes with replacement. Let, $X_i$ be a random variable defined as follows.
 $$ X_i = \begin{cases}
           1,& \text{if } x \text{ is at first position in the } i^{th} \text{ sample}\\
           0,& \text{else}
          \end{cases}
 $$
 The estimate of the plurality score of $x$ be $\hat{s}(x)$. Then $\hat{s}(x) = \frac{n}{l}\sum_{i=1}^{l}X_i$. Let $s(x)$ be the actual plurality score of $x$. Then we have following,
 $$ E[X_i] = \frac{s(x)}{n}, E[ \hat{s}(x) ] = \frac{n}{l} \sum_{i=1}^{l} E[ X_i ] = s(x)$$
 By Chernoff bound, we have the following,
 $$ \Pr[ |\hat{s}(x) - s(x)| > \epsilon n ] \le \frac{2}{\exp\{\epsilon^2 l n/3s(x)\}}$$
 By union bound, we have the following,
 \begin{eqnarray*}
  \Pr[ \exists x\in \mathcal{C}, |\hat{s}(x) - s(x)| > \epsilon n ] &\le& \sum_{x\in \mathcal{C}} \frac{2}{\exp\{\epsilon^2 ln/3s(x)\}}\\
 &\le& \frac{2}{\exp\{\epsilon^2 l/3\}}
 \end{eqnarray*}
 The last line follows from \Cref{lem:funmax}. Notice that, we do not need the random variables $\hat{s}(x)$ and $\hat{s}(y)$ to be independent for any two candidates $x$ and $y$. Hence, we can use the same $l$ sample votes to estimate $\hat{s}(x)$ for every candidate $x$.
 
 Now, let $y$ and $z$ be the two candidates that go to the second round.
 $$ Y_i = \begin{cases}
           1,& \text{if } y\succ z \text{ in the } i^{th} \text{ sample}\\
           -1,& \text{else}
          \end{cases}
 $$
 The estimate of $D(y,z)$ be $\hat{D}(y,z)$. Then $\hat{D}(y,z) = \frac{n}{l}\sum_{i=1}^{l}Y_i$. Then we have following,
 $$ E[Y_i] = \frac{D(y,z)}{n}, E[ \hat{D}(y,z) ] = \frac{n}{l} \sum_{i=1}^{l} E[ Y_i ] = D(y,z)$$
 By Chernoff bound, we have the following,
 $$ \Pr[ |\hat{D}(y,z) - D(y,z)| > \epsilon n ] \le \frac{2}{\exp\{\epsilon^2 l/3\}}$$
 Let $A$ be the event that $\forall x\in \mathcal{C}, |\hat{s}(x) - s(x)| \le \epsilon n$ and $ |\hat{D}(y,z) - D(y,z)| \le \epsilon n$. Now we have,
 $$\Pr[A ] \ge 1 - (\frac{2}{\exp\{\epsilon^2 l/3\}} + \frac{2}{\exp\{\epsilon^2 l/3\}})$$
 Since we do not need independence among the random variables $\hat{s}(a)$, $\hat{s}(b)$, $\hat{D}(w,x)$, $\hat{D}(y,z)$ for any candidates $a, b, w, x, y,$ and $z$, we can use the same $l$ sampled votes. Now, from \Cref{lem:runoff}, if $|\hat{s}(x) - s(x)| \le \frac{\epsilon n}{5}$ for every candidate $x$ and $|\hat{D}(y,z) - D(y,z)| \le \frac{\epsilon n}{5}$ for every candidates $y$ and $z$, then the plurality with runoff winner of the sampled votes coincides with the actual runoff winner. The above event happens with probability at least $1-\delta$ by choosing an appropriate $ l = O(\frac{\log \frac{1}{\delta}}{\epsilon^2})$.\qed
\end{proof}

\subsection{STV Voting Rule}

Now we move on the STV voting rule. The 
following lemma provides an upper bound on the number of votes that need to be changed to make some arbitrary candidate win the election. 
More specifically, given a sequence of $m$ candidates $\{x_i\}_{i=1}^m$ with $x_m$ not being the winner, the lemma below proves an upper bound on the number of number of votes that need to be modified such that the candidate $x_i$ gets eliminated at the $i^{th}$ round in the STV voting rule.

\begin{lemma}\label{lem:stv}
 Suppose $\mathcal{V}$ be a set of votes and $w$ be the winner of a STV election. Consider the following chain with candidates $x_1\ne x_2\ne \ldots \ne x_m$ and $ x_m \ne w $.
 $$ \mathcal{C}\supset \mathcal{C}\setminus\{x_1\}\supset \mathcal{C}\setminus\{x_1,x_2\}\supset \ldots \supset\{x_m\} $$
 Let $s_{\mathcal{V}}(A,x)$ be the plurality score of a candidate $x$ when all the votes in $\mathcal{V}$ are restricted to the set of candidates $A\subset \mathcal{C}$. Let us define $\mathcal{C}_{-i} = \mathcal{C}\setminus \{x_1, \ldots, x_i\}$ and $s^*_{\mathcal{V}}(A) := \min_{x\in A} \{s_{\mathcal{V}}(A,x)\}$. Then, we have the following.
 $$ \sum_{i=0}^{m-1} \left(s_{\mathcal{V}}\left({\mathcal{C}_{-i}},
     x_{i+1}\right) -
   s^*_{\mathcal{V}}\left({\mathcal{C}_{-i}}\right)\right) \ge
 \textsf{MOV}$$
\end{lemma}

\begin{proof}
 We will show that by changing $ \sum_{i=0}^{m-1} \left(s_{\mathcal{V}}\left({\mathcal{C}_{-i}}, x_{i+1}\right) - s^*_{\mathcal{V}}\left({\mathcal{C}_{-i}}\right)\right) $ votes, we can make the candidate $x_m$ winner. If $x_1$ minimizes $s_{\mathcal{V}}(\mathcal{C},x)$ over $x\in \mathcal{C}$, then we do not change anything and define $\mathcal{V}_1 = \mathcal{V}$. Otherwise, there exist $s_{\mathcal{V}}(\mathcal{C},x_1) - s^*_{\mathcal{V}}(\mathcal{C})$ many votes of following type.
 $$ x_1\succ a_1\succ a_2\succ \ldots \succ a_{m-1}, a_i\in \mathcal{C}, \forall 1\le i\le m-1 $$
 We replace $s_{\mathcal{V}}(\mathcal{C},x_1) - s^*_{\mathcal{V}}(\mathcal{C})$ many votes of the above type by the votes as follows.
 $$ a_1\succ x_1\succ a_2\succ \ldots \succ a_{m-1} $$
 Let us call the new set of votes by $\mathcal{V}_1$. We claim that, $s_{\mathcal{V}}(\mathcal{C}\setminus {x_1}, x) = s_{\mathcal{V}_1}(\mathcal{C}\setminus {x_1}, x)$ for every candidate $x\in \mathcal{C}\setminus\{x_1\}$. Fix any arbitrary candidate $x\in \mathcal{C}\setminus\{x_1\}$. The votes in $\mathcal{V}_1$ that are same as in $\mathcal{V}$ contributes same quantity to both side of the equality. Let $v$ be a vote that has been changed as described above. If $x = a_1$ then, the vote $v$ contributes one to both sides of the equality. If $x \ne a_1$, then the vote contributes zero to both sides of the equality. Hence, we have the claim. We repeat this process for $(m-1)$ times. Let $\mathcal{V}_i$ be the set of votes after the candidate $x_i$ gets eliminated. Now, in the above argument, by replacing $\mathcal{V}$ by $\mathcal{V}_{i-1}$, $\mathcal{V}_1$ by $\mathcal{V}_i$, the candidate set $\mathcal{C}$ by $\mathcal{C}\setminus \{x_1, \ldots, x_{i-1}\}$, and the candidate $x_1$ by the candidate $x_i$, we have the 
following.
 $$ s_{\mathcal{V}_{i-1}}(\mathcal{C}_{-i}, x) = s_{\mathcal{V}_i}(\mathcal{C}_{-i}, x) \forall x\in \mathcal{C}\setminus\{x_1, \ldots, x_i\}$$
 Hence, we have the following.
 $$ s_{\mathcal{V}}(\mathcal{C}_{-i}, x) = s_{\mathcal{V}_i}(\mathcal{C}_{-i}, x) \forall x\in \mathcal{C}\setminus\{x_1, \ldots, x_i\}$$
 In the above process, the total number of votes that are changed is $ \sum_{i=0}^{m-1} \left(s_{\mathcal{V}}\left({\mathcal{C}_{-i}}, x_{i+1}\right) - s^*_{\mathcal{V}}\left({\mathcal{C}_{-i}}\right)\right) $.\qed
\end{proof}

\begin{theorem}\label{thm:stv}
 There is a $(\epsilon, \delta)$-winner determination algorithm for the STV voting rule with sample complexity $ O(\frac{m^2(m+\log \frac{1}{\delta})}{\epsilon^2})$.
\end{theorem}

\begin{proof}
 We sample $l$ votes uniformly at random from the set of votes with replacement and output the STV winner of those $l$ votes say $w^{\prime}$ as the winner of the election. Let, $w$ be the winner of the election. We will show that there exist $l =  O(\frac{m^2(m+\log \frac{1}{\delta})}{\epsilon^2})$ for which $w=w^{\prime}$ with probability at least $1 - \delta$. Let $A$ be an arbitrary subset of candidates and $x$ be any candidate in $A$. Let us define a random variables $X_i, 1\le i\le l$ as follows.
 $$ X_i = \begin{cases}
           1,& \text{if } x \text{ is at top } i^{th} \text{ sample when restricted to } A\\
           0,& \text{else}
          \end{cases}
 $$
 Define another random variable $\hat{s}_{\mathcal{V}}(A,x) := \sum_{i=1}^l X_i$. Then we have, $E[\hat{s}_{\mathcal{V}}(A,x)] = s_{\mathcal{V}}(A,x)$. Now, using Chernoff bound, we have the following,
 $$ \Pr[ |\hat{s}_{\mathcal{V}}(A,x) - s_{\mathcal{V}}(A,x)| > \frac{\epsilon n}{m} ] \le \frac{2}{\exp\{\frac{\epsilon^2 l}{3m^2}\}}$$
 Let $E$ be the event that $\exists A\subset \mathcal{C} \text{ and } \exists x\in A, |\hat{s}_{\mathcal{V}}(A,x) - s_{\mathcal{V}}(A,x)| > \frac{\epsilon n}{m}$. By union bound, we have,
 \begin{eqnarray*}
  \Pr[ \bar{E} ] &\ge& 1 - \frac{m2^{m+1}}{\exp\{\frac{\epsilon^2 l}{3m^2}\}}
 \end{eqnarray*}
 The rest of the proof follows by an argument analogous to the proof of \Cref{thm:gen} using \Cref{lem:stv}.\qed
\end{proof}

\section{Conclusion}\label{sec:con}

In this work, we introduced the $(\eps,\delta)$-winner determination
problem and showed (often tight) bounds for the sample complexity for
many common voting rules. Besides closing the remaining gaps in the
bounds, here are a few open directions to pursue in
the future:
\begin{itemize}
\item
Is there an axiomatic characterization of the voting rules for which
the sample complexity is independent of $m$ and $n$? We note that a
similar problem in graph property testing was the subject of intense
study \citep{alon2006acomb, borgs2006graph}. 
\item
Specifically for scoring rules, is the sample complexity determined by
some natural property of the score vector, such as its sparsity? 
\item
Is it worthwhile for the algorithm to elicit only part of the vote
from each sampled voter instead of the full vote? As mentioned in the
Introduction, vote elicitation is a well-trodden area, but as far
as we know, it has not been studied how assuming a margin of victory
can change the number of queries.
\item
How can knowledge of a social network on the voters be used to
minimize the number of samples made? Some initial progress in this
direction has been made by Dhamal and Narahari \citep{dhamal2013scalable} and by Agrawal
and Devanur (private communication).

\end{itemize}

\bibliographystyle{apalike}

\bibliography{bwd}

\begin{thebibliography}{}

\bibitem[Alon et~al., 2006]{alon2006acomb}
Alon, N., Fischer, E., Newman, I., and Shapira, A. (2006).
\newblock A combinatorial characterization of the testable graph properties:
  it's all about regularity.
\newblock In {\em Proceedings of the 38th Annual ACM Symposium on Theory of
  Computing (STOC)}, pages 251--260.

\bibitem[Bachrach et~al., 2010]{bachrach2010probabilistic}
Bachrach, Y., Betzler, N., and Faliszewski, P. (2010).
\newblock Probabilistic possible winner determination.
\newblock In {\em International Conference on Artificial Intelligence (AAAI)},
  volume~10, pages 697--702.

\bibitem[Bar-Yossef, 2003]{bar2003sampling}
Bar-Yossef, Z. (2003).
\newblock Sampling lower bounds via information theory.
\newblock In {\em Proceedings of the 35th Annual ACM Symposium on Theory of
  Computing (STOC)}, pages 335--344. ACM.

\bibitem[Bar-Yossef et~al., 2001]{bar2001sampling}
Bar-Yossef, Z., Kumar, R., and Sivakumar, D. (2001).
\newblock Sampling algorithms: lower bounds and applications.
\newblock In {\em Proceedings of the 33rd Annual ACM Symposium on Theory of
  Computing (STOC)}, pages 266--275. ACM.

\bibitem[Bartholdi~III et~al., 1989]{bartholdi1989voting}
Bartholdi~III, J., Tovey, C.~A., and Trick, M.~A. (1989).
\newblock Voting schemes for which it can be difficult to tell who won the
  election.
\newblock {\em Social Choice and Welfare (SCW)}, 6(2):157--165.

\bibitem[Boldi et~al., 2009]{boldi2009voting}
Boldi, P., Bonchi, F., Castillo, C., and Vigna, S. (2009).
\newblock Voting in social networks.
\newblock In {\em Proceedings of the 18th ACM Conference on Information and
  Knowledge Mmanagement}, pages 777--786. ACM.

\bibitem[Borgs et~al., 2006]{borgs2006graph}
Borgs, C., Chayes, J.~T., Lov{\'a}sz, L., S{\'o}s, V.~T., Szegedy, B., and
  Vesztergombi, K. (2006).
\newblock Graph limits and parameter testing.
\newblock In {\em Proceedings of the 38th Annual ACM Symposium on Theory of
  Computing (STOC)}, pages 261--270.

\bibitem[Boutilier et~al., 2014]{boutilier2014robust}
Boutilier, C., Lang, J., Oren, J., and Palacios, H. (2014).
\newblock Robust winners and winner determination policies under candidate
  uncertainty.
\newblock In {\em Proceedings of the International Conference on Artificial
  Intelligence (AAAI)}.

\bibitem[Canetti et~al., 1995]{canetti1995lower}
Canetti, R., Even, G., and Goldreich, O. (1995).
\newblock Lower bounds for sampling algorithms for estimating the average.
\newblock {\em Information Processing Letters (IPL)}, 53(1):17--25.

\bibitem[Conitzer, 2009]{conitzer2009eliciting}
Conitzer, V. (2009).
\newblock Eliciting single-peaked preferences using comparison queries.
\newblock {\em Journal of Artificial Intelligence Research (JAIR)},
  35:161--191.

\bibitem[Conitzer and Sandholm, 2002]{conitzer2002vote}
Conitzer, V. and Sandholm, T. (2002).
\newblock Vote elicitation: Complexity and strategy-proofness.
\newblock In {\em Eighteenth International Conference on Artificial
  Intelligence (AAAI)}, pages 392--397, Menlo Park, CA, USA. American
  Association for Artificial Intelligence.

\bibitem[Dey and Bhattacharyya, 2015]{deysampling}
Dey, P. and Bhattacharyya, A. (2015).
\newblock Sample complexity for winner prediction in elections.
\newblock In {\em Proceeding of the 14th International Conference on Autonomous
  Systems and Multiagent Systems (AAMAS-15)}.

\bibitem[Dhamal and Narahari, 2013]{dhamal2013scalable}
Dhamal, S. and Narahari, Y. (2013).
\newblock Scalable preference aggregation in social networks.
\newblock In {\em First AAAI Conference on Human Computation and Crowdsourcing
  (HCOMP)}.

\bibitem[Ding and Lin, 2012]{ding2012voting}
Ding, N. and Lin, F. (2012).
\newblock Voting with partial information: Minimal sets of questions to decide
  an outcome.
\newblock In {\em Proc. Fourth International Workshop on Computational Social
  Choice (COMSOC-2012), Krak{\'o}w, Poland}.

\bibitem[Doucette et~al., 2014]{doucetteapproximate}
Doucette, J.~A., Larson, K., and Cohen, R. (2014).
\newblock Approximate winner selection in social choice with partial
  preferences.
\newblock In {\em Proceedings of the 12th International Conference on
  Autonomous Agents and Multiagent Systems (AAMAS)}. International Foundation
  for Autonomous Agents and Multiagent Systems.

\bibitem[Dvoretzky et~al., 1956]{dvoretzky1956asymptotic}
Dvoretzky, A., Kiefer, J., and Wolfowitz, J. (1956).
\newblock Asymptotic minimax character of the sample distribution function and
  of the classical multinomial estimator.
\newblock {\em The Annals of Mathematical Statistics}, pages 642--669.

\bibitem[Ephrati and Rosenschein, 1991]{ephrati1991clarke}
Ephrati, E. and Rosenschein, J. (1991).
\newblock The {C}larke tax as a consensus mechanism among automated agents.
\newblock In {\em Proceedings of the Ninth International Conference on
  Artificial Intelligence (AAAI)}, pages 173--178.

\bibitem[Falik et~al., 2012]{falik2012coalitions}
Falik, D., Meir, R., and Tennenholtz, M. (2012).
\newblock On coalitions and stable winners in plurality.
\newblock In {\em 8th International Workshop on Internet and Network Economics
  (WINE)}, pages 256--269. Springer.

\bibitem[Hazon et~al., 2008]{hazon2008evaluation}
Hazon, N., Aumann, Y., Kraus, S., and Wooldridge, M. (2008).
\newblock Evaluation of election outcomes under uncertainty.
\newblock In {\em Proceedings of the 7th International Conference on Autonomous
  Agents and Multiagent Systems (AAMAS)}, pages 959--966. International
  Foundation for Autonomous Agents and Multiagent Systems.

\bibitem[Hemaspaandra et~al., 1997]{hemaspaandra1997exact}
Hemaspaandra, E., Hemaspaandra, L.~A., and Rothe, J. (1997).
\newblock Exact analysis of dodgson elections: Lewis carroll's 1876 voting
  system is complete for parallel access to np.
\newblock {\em Journal of the ACM (JACM)}, 44(6):806--825.

\bibitem[Hemaspaandra et~al., 2005]{hemaspaandra2005complexity}
Hemaspaandra, E., Spakowski, H., and Vogel, J. (2005).
\newblock The complexity of kemeny elections.
\newblock {\em Theoretical Computer Science (TCS)}, 349(3):382--391.

\bibitem[Kullback and Leibler, 1951]{kullback1951information}
Kullback, S. and Leibler, R.~A. (1951).
\newblock On information and sufficiency.
\newblock {\em The Annals of Mathematical Statistics}, pages 79--86.

\bibitem[Lin, 1991]{lin1991divergence}
Lin, J. (1991).
\newblock Divergence measures based on the shannon entropy.
\newblock {\em IEEE Transactions on Information Theory}, 37(1):145--151.

\bibitem[Lu and Boutilier, 2011a]{lu2011robust}
Lu, T. and Boutilier, C. (2011a).
\newblock Robust approximation and incremental elicitation in voting protocols.
\newblock In {\em Proceedings of the 22nd International Joint Conference on
  Artificial Intelligence (IJCAI)}, volume~22, page 287.

\bibitem[Lu and Boutilier, 2011b]{lu2011vote}
Lu, T. and Boutilier, C. (2011b).
\newblock Vote elicitation with probabilistic preference models: Empirical
  estimation and cost tradeoffs.
\newblock In {\em Algorithmic Decision Theory}, pages 135--149. Springer.

\bibitem[Lu and Boutilier, 2013]{lu2013multi}
Lu, T. and Boutilier, C. (2013).
\newblock Multi-winner social choice with incomplete preferences.
\newblock In {\em Proceedings of the 23rd International Joint Conference on
  Artificial Intelligence (IJCAI)}, pages 263--270. AAAI Press.

\bibitem[Oren et~al., 2013]{oren2013efficient}
Oren, J., Filmus, Y., and Boutilier, C. (2013).
\newblock Efficient vote elicitation under candidate uncertainty.
\newblock In {\em Proceedings of the 23rd International Joint Conference on
  Artificial Intelligence (IJCAI)}, pages 309--316. AAAI Press.

\bibitem[Pennock et~al., 2000]{pennock2000social}
Pennock, D.~M., Horvitz, E., and Giles, C.~L. (2000).
\newblock Social choice theory and recommender systems: Analysis of the
  axiomatic foundations of collaborative filtering.
\newblock In {\em Proceedings of the 17th International Conference on
  Artificial Intelligence (AAAI)}.

\bibitem[Rodriguez et~al., 2007]{rodriguez2007smartocracy}
Rodriguez, M.~A., Steinbock, D.~J., Watkins, J.~H., Gershenson, C., Bollen, J.,
  Grey, V., and Degraf, B. (2007).
\newblock Smartocracy: Social networks for collective decision making.
\newblock In {\em In 40th Annual Hawaii International Conference on Systems
  Science (HICSS'07). Waikoloa}, pages 90--90. IEEE.

\bibitem[Ron, 2001]{ron2001property}
Ron, D. (2001).
\newblock Property testing.
\newblock {\em Combinatorial Optimization-Dordrecht}, 9(2):597--643.

\bibitem[Shiryaev et~al., 2013]{shiryaev2013elections}
Shiryaev, D., Yu, L., and Elkind, E. (2013).
\newblock On elections with robust winners.
\newblock In {\em Proceedings of the 12th International Conference on
  Autonomous Agents and Multiagent Systems (AAMAS)}, pages 415--422.
  International Foundation for Autonomous Agents and Multiagent Systems.

\bibitem[Waggoner, 2015]{Waggoner2015}
Waggoner, B. (2015).
\newblock Lp testing and learning of discrete distributions.
\newblock In {\em Proceedings of the 2015 Conference on Innovations in
  Theoretical Computer Science}, ITCS '15, pages 347--356, New York, NY, USA.
  ACM.

\bibitem[Xia, 2012]{xia2012computing}
Xia, L. (2012).
\newblock Computing the margin of victory for various voting rules.
\newblock In {\em Proceedings of the 13th ACM Conference on Electronic Commerce
  (EC)}, pages 982--999. ACM.

\end{thebibliography}

\end{document}